\documentclass[copyright,creativecommons]{eptcs}
\providecommand{\event}{RTRTS 2010} 

\newif\iflong\longfalse  

\newif\ifworking\workingfalse 

\usepackage{breakurl}             

\usepackage{amsmath}
\usepackage{amsfonts}
\usepackage{amssymb}

\usepackage{tikz}
\usetikzlibrary{arrows}

\newcommand{\finally}{\Diamond}
\newcommand{\globally}{\Box}
\newcommand{\until}{\textit{U}}
\newcommand{\weakuntil}{\textit{W}}

\newcommand{\reals}{\mathbb{R}}
\newcommand{\nats}{\mathbb{N}}

\usepackage[latin1]{inputenc}
\usepackage{latexsym}
\usepackage{alltt}
\usepackage{shortvrb}  
\usepackage{dsfont}    

\newcommand{\clockvalue}{\texttt{clock}}
\newcommand{\clockvaluesem}{\textit{clock}}

\newcommand{\cpq}{c_{\textit{BR}}}
\newcommand{\cp}{c_{\textit{MS}}}

\newcommand{\transBR}{\textit{BR}}
\newcommand{\transMS}{\textit{MS}}
\newcommand{\LBR}{L_{\Pi}}
\newcommand{\tildeLBR}{\widetilde{L}_{\Pi}}
\newcommand{\LMS}{L_{\Pi}}
\newcommand{\tildeLMS}{\widetilde{L}_{\Pi}}

\newcommand{\tnull}{t_0}

\newcommand{\buntil}{\;U_{\leq r}\;}

\newcommand{\beventually}{\Diamond_{\leq r}\;}

\newcommand{\balways}{\Box_{\leq r}\;}

\newcommand{\attone}{\mbox{$att_1$}}
\newcommand{\attn}{\mbox{$att_n$}}
\newcommand{\sone}{\mbox{$s_1$}}
\newcommand{\sn}{\mbox{$s_n$}}

\usepackage{amsthm}

\newtheorem{theorem}{Theorem}[section]

\newtheorem{lem}[theorem]{Lemma}

\newtheorem{definition}[theorem]{Definition}

\ifworking
\newcounter{todo}
\newcommand{\commentary}[1]{{\stepcounter{todo}\marginpar{\footnotesize #1}}}
\else
\newcommand{\commentary}[1]{}
\fi

\title{Model Checking Classes of Metric LTL Properties of Object-Oriented Real-Time Maude Specifications}

\author{%
Daniela Lepri
\institute{University of Oslo, Norway}
\and
Peter Csaba {\"O}lveczky
\institute{University of Oslo, Norway}
\and
Erika \'Abrah\'am
\institute{RWTH Aachen University, Germany}
}

\def\titlerunning{Model Checking 
MTL Properties of Object-Oriented Real-Time Maude Specifications}
\def\authorrunning{D. Lepri, P.\:Cs. \"Olveczky, and E. \'Abrah\'am}

\begin{document}

\MakeShortVerb{\@}     

\maketitle

\begin{abstract}
This paper presents a  transformational approach for model
checking two important classes of metric temporal logic (MTL) properties, 
namely, bounded response
and minimum separation, for  non-hierarchical object-oriented Real-Time Maude
specifications. 
We prove the correctness of our model checking algorithms, which terminate
under reasonable non-Zeno-ness assumptions 
when the reachable state space is finite.
These new model checking features have been
 integrated into Real-Time Maude, and are used to 
 analyze a network of medical devices and  a 4-way traffic  intersection system.
\end{abstract}
\section{Introduction}

Real-Time Maude~\cite{journ-rtm}  is a formal specification language and 
a high-performance  simulation and model checking 
tool that extends the rewriting-logic-based
Maude system~\cite{maude-book} to support the formal specification
and analysis of  \emph{real-time} systems. 
Real-Time Maude differs from timed-automaton-based tools, such as 
{\sc Uppaal}~\cite{uppaalTutorial} and {\sc Kronos}~\cite{kronos}, 
by emphasizing ease and expressiveness of specification
over   algorithmic decidability of key properties. In particular, 
Real-Time Maude supports the definition of any computable data type,
unbounded data structures, different  communication models, and so on.

Because of its expressiveness, Real-Time Maude has been successfully
applied to a wide range of advanced state-of-the-art applications that are beyond the pale
of timed automata, including 
the OGDC density control~\cite{ogdc-tcs} and LMST topology control~\cite{mike-wsn} 
protocols for wireless sensor networks,   
the CASH scheduling algorithm  with  capacity sharing
features that require unbounded queues~\cite{fase06}, the AER/NCA
active networks multicast protocol~\cite{aer-journ}, and the NORM multicast protocol
developed by the IETF~\cite{norm-paper}.  Real-Time Maude's natural model of time, together with its
expressiveness, also makes it ideal as a semantic framework in which
 real-time modeling languages can be given a formal semantics; such languages
then also get Real-Time Maude's
formal analysis capabilities essentially for free. Languages with a  Real-Time Maude 
semantics include: a timed extension of the Actor model~\cite{rtactors-in-rtmaude}, 
the Orc web services orchestration language~\cite{musab-orc07}, 
a language 
developed at DoCoMo laboratories for handset applications~\cite{musab-fase09},
a behavioral subset of the avionics standard AADL~\cite{aadl-fmoods},
the visual model transformation language e-Motions~\cite{emotions-wrla10},
 real-time model transformations in MOMENT2~\cite{fase10},
and a  subset of  Ptolemy II  discrete-event models~\cite{icfem09}.   

Real-Time Maude is particularly suitable to model real-time systems in
an \emph{object-oriented} style, and the paper~\cite{journ-rtm} identifies 
some useful specification techniques for  object-oriented 
 real-time systems. All the concrete applications mentioned above, and many of
the language semantics applications, are specified in an object-oriented 
way using those techniques. 

Real-Time Maude provides a spectrum of analysis methods, including
simulation through timed rewriting, untimed temporal logic model
checking, and (unbounded or time-bounded) search for reachability
analysis. However, up to know, Real-Time Maude has lacked the ability
to model check \emph{timed} (or \emph{metric}) temporal logic
properties. Such properties are obviously very important in many
real-time systems. For example, in case of an accident the airbag must
not just inflate \emph{eventually}, but within very tight time bounds.
For timed automata, such metric temporal logic model checking is
decidable\footnote{for finite behaviour, see, e.g., \cite{bouyer:phd}}, and
implemented in the {\sc Kronos} tool~\cite{kronos}.
For the much more expressive Real-Time Maude formalism,
supporting metric temporal logic checking, is obviously a much harder task.

This paper reports on our first attempts at providing metric temporal logic
model checking for Real-Time Maude. 
We have taken the following pragmatic choices:
\begin{enumerate}
\item Supporting the model checking of 
only a few classes of metric temporal logic properties, namely, the ones that
were needed in the above-mentioned applications. These 
properties are: 
\begin{itemize} \item \emph{Bounded response}: each $p$-state must be
  followed by a $q$-state within time $r$ (where $p$ and $q$ are state
  propositions).  One example of a bounded response property is
  ``whenever the ventilator assisting the patient's breathing is
  turned off, it must be turned on within 5 seconds''.
\item \emph{Minimum separation}: there must be at least time $r$
  between two non-consecutive $p$-states. For example, ``the
  ventilator should be turned on continuously for at least two minutes
  between two pauses.''
\end{itemize}
\item 
Supporting  such model checking only for flat object-oriented models
specified according to the guidelines mentioned above. But as already said,  
this class of systems includes all
the concrete Real-Time Maude 
applications listed above. 
\end{enumerate}
What is gained by restricting the classes of systems and properties is
\emph{efficiency}.  Instead of implementing the model checking
algorithms from scratch, we pursue a \emph{transformational} approach,
where we take advantage of Maude's high performance analysis commands
and transform a metric model checking problem $\mathcal{R}, L, t_o
\models \phi$ into a problem $\widetilde{\mathcal{R}}, \widetilde{L},
\widetilde{t_o} \models \widetilde{\phi}$ that can be analyzed by
Real-Time Maude's efficient search and LTL model checking commands. Our transformations add a
clock which measures, respectively, the time since the earliest
$(p\wedge \neg q)$-state that has not been followed by a $q$-state
(for bounded response) and the last time since we saw a $p$-state (for
minimum separation).  An important property is that -- under
reasonable time-divergence assumptions about the executions with the
selected time sampling strategy -- if the original reachable state
space is finite, then the model checking commands are guaranteed to
terminate.  Furthermore, our model checking commands
are semi-decision procedures for the invalidity of the metric properties
for time-diverging systems. 
The transformations have been implemented in Real-Time
Maude and the corresponding model checking commands have been made
available in the tool.  We have applied the new commands on two case
studies, one on the safe interoperation of medical
devices~\cite{phuket08} and one on a fault-tolerant controller for
traffic lights in an intersection~\cite{traffic-light}.

We prove the correctness of the transformation under reasonable
assumptions, such as the real-time rewrite theory
being
\emph{tick-invariant}~\cite{wrla06}. Since real-time rewrite
theories do not have a ``region-automaton''-like discrete quotient,
for dense time Real-Time Maude uses \emph{time sampling strategies} to
execute the tick rules. That is, in model checking analyses for
dense-time models, only a subset of all possible behaviors are
analyzed. Therefore, Real-Time Maude analyses are in general not
(both) sound and complete; however, for object-oriented specifications
we have identified easily checkable conditions that guarantee
soundness and completeness of our analyses also for dense-time
systems~\cite{wrla06}.

This paper is organized as follows. Section~\ref{sec:preliminaries}
introduces Real-Time Maude and metric temporal~logic.
Section~\ref{sec:implementation} presents the properties that we
address and the corresponding transformations, whose correctness is
proved in Section~\ref{sec:proof}.  Section~\ref{sec:case_studies}
shows two case studies of metric temporal logic model checking in
Real-Time Maude.  Section~\ref{sec:related} discusses related work,
and Section~\ref{sec:conclusion} gives some
concluding remarks.
\section{Preliminaries}
\label{sec:preliminaries}

\subsection{Real-Time Maude}
\label{sec:rtm}

In Real-Time Maude~\cite{journ-rtm},
real-time systems are modeled by a set of \emph{equations} and 
\emph{rewrite rules}. The rewrite rules are divided into \emph{instantaneous}
rules, that model changes that are assumed to take zero time, and 
\emph{tick} rules that model time advance. 
Formally, a Real-Time Maude \emph{timed module} specifies a  \emph{real-time
rewrite theory}~\cite{OlvMesTCS}  of the form
 $\mathcal{R}=(\Sigma, E, \mathit{IR}, \mathit{TR})$, where:
\begin{itemize}
\item $(\Sigma, E)$ is a \emph{membership equational
    logic}~\cite{maude-book} theory with $\Sigma$ a
  signature\footnote{That is, $\Sigma$ is a set of declarations of
    \emph{sorts}, \emph{subsorts}, and \emph{function symbols}.} and
  $E$ a set of {\em confluent and terminating conditional equations}.
  $(\Sigma, E)$ specifies the system's state space as an algebraic
  data type, and must contain a specification of a sort @Time@
  modeling the (discrete or dense) time domain. We denote by
  $\mathbb{T}_{\cal{R},\texttt{s}}$ all ground terms of sort
  \texttt{s}.

\item  $\mathit{IR}$
is a set of (possibly conditional) \emph{labeled instantaneous (rewrite)
  rules} specifying  
the system's \emph{instantaneous} (i.e., zero-time) local transitions,  written 
$@crl [@l@] : @t@ => @t'@ if @cond$,
where  $l$ is a \emph{label}. 
Such a rule specifies a \emph{one-step transition}
from an instance of $t$ to the corresponding 
instance 
of $t'$. 
 The rules are applied \emph{modulo} the
equations~$E$.\footnote{$E$  is a union $E'\cup
  A$, where $A$ is a set of equational axioms such as associativity,
  commutativity, and identity, so that deduction is performed \emph{modulo}
  $A$. Operationally,  a
term is reduced to its
$E'$-normal form modulo $A$ before any rewrite rule is applied.} 

\item $\mathit{TR}$ is a set of \emph{tick (rewrite) rules}, written
  with syntax  
\vspace{-1mm}
\begin{alltt}
 crl [\(l\)] : \texttt{\char123}\(t\)\texttt{\char125} => \texttt{\char123}\(t'\)\texttt{\char125} in time \(\tau\) if \(cond\) .
\end{alltt}
\vspace{-1mm}
that model time elapse.  @{_}@ is a 
built-in
 constructor of  sort \texttt{GlobalSystem}, and
$\tau$ is a term of sort @Time@ that denotes the \emph{duration}
of the rewrite.
\end{itemize}
The initial state must be a ground term of sort @GlobalSystem@ and 
 must be reducible to a term of
the form @{@$t$@}@ using the equations in the specification. 
The form of the tick rules ensures that time advances uniformly in the whole system.


Following~\cite{OlvMesTCS}, we write $t \stackrel{r}{\rightarrow}
t'$ when $t$ can be rewritten into $t'$ in time $r$ by a
\emph{one-step rewrite}.  Note that instantaneous steps have 
duration $0$. A \emph{(timed) path} $\pi$ in $\cal R$ is an infinite
sequence
  $$
  \pi = t_0 \stackrel{r_0}{\rightarrow} t_1
  \stackrel{r_1}{\rightarrow} t_2\ldots
  $$
  such that either
  \begin{itemize}
  \item for all $i\in\nats$, $t_i\xrightarrow{r_i} t_{i+1}$ is a one-step
    rewrite in $\mathcal{R}$;
    or
  \item there exists a $k\in\nats$ such that $\;t_i\xrightarrow{r_i}
    t_{i+1}$ is a one-step rewrite in $\mathcal{R}$ for all $0\leq i <
    k$, there is no one-step rewrite from $t_k$ in $\mathcal{R}$,
    and $t_j = t_k$ and $r_{j-1} = 0$ for each $j>k$.
  \end{itemize}
We denote by $Paths(\mathcal{R})_{t_0}$
  the set of all timed paths of $\cal R$  starting in $t_0$.   We call a path $ \pi = t_0 \stackrel{r_0}{\rightarrow} t_1
  \stackrel{r_1}{\rightarrow} t_2\ldots$ \emph{time-divergent} iff for
  all $r\in\reals$ there is an $i\in\nats$ such that $\sum_{k=0}^i r_k
  > r$. Paths that are not time-divergent are called
  \emph{time-convergent}.
  We define $\pi^k = t_k
  \stackrel{r_{k}}{\rightarrow} t_{k+1}
  \stackrel{r_{k+1}}{\rightarrow} \ldots$.  A term $t'$ is
  \emph{reachable} from $t_0$ in $\cal R$ in time $r$ iff there is a
  path $\pi = t_0 \stackrel{r_0}{\rightarrow} \ldots
  \stackrel{r_{k-1}}{\rightarrow} t_k \ldots$ with $t_k = t'$ and
  $r=\sum_{i=0}^{k-1}r_i$.

\smallskip

  The Real-Time Maude syntax is fairly intuitive; we refer
  to~\cite{maude-book} for a detailed description. For example, a
  function symbol $f$ is declared with the syntax \texttt{op }$f$ @:@
  $s_1$ \ldots $s_n$ @->@ $s$, where $s_1\:\ldots\:s_n$ are the sorts
  of its arguments, and $s$ is its (value) \emph{sort}. Equations are
  written with syntax @eq@ $t$ @=@ $t'$, and @ceq@ $t$ @=@ $t'$ @if@
  \emph{cond} are conditional equations. The mathematical variables in
  such statements are declared with the keywords {\tt var} and {\tt
    vars}.


In \emph{object-oriented} Real-Time Maude modules, a \emph{class} declaration
\begin{alltt}
 class \(C\) | \(\attone\) : \(\sone\), \dots , \(\attn\) : \(\sn\) .
\end{alltt}
declares a class $C$ with attributes $att_1$ to $att_n$ of
sorts $s_1$ 
to $s_n$, respectively. An {\em object\/} of class $C$ in a   state is
represented as a term
$@<@\: O : C \mid att_1: val_1, ... , att_n: val_n\:@>@$
of sort @Object@, where $O$, of sort @Oid@,  is the
object's
\emph{identifier}, and where $val_1$ to 
$val_n$ are the current values of the attributes $att_1$ to
$att_n$, respectively.
 In a \emph{concurrent} object-oriented
system, the 
 state
 is a term of 
sort @Configuration@. It  has 
the structure of a  \emph{multiset} made up of objects and messages.
Multiset union for configurations is denoted by a juxtaposition
operator (empty
syntax) that is declared associative and commutative, so that rewriting is 
\emph{multiset
rewriting} supported directly in Real-Time Maude.
The dynamic behavior of concurrent
object systems is axiomatized by specifying  its 
transition patterns by  rewrite rules. For example, 
  the rule  

{\small
\begin{alltt}
  rl [l] : m(O,w) < O : C | a1 : 0, a2 : y, a3 : w >  =>
                  < O : C | a1 : T, a2 : y, a3 : y + w > dly(m'(O'),x) .  
\end{alltt}
}

\noindent defines a parametrized family of transitions
(one for each substitution instance), which can be applied
whenever the attribute @a1@ of an 
object {\tt O} of class @C@ has the value @0@,  with the effect of altering
the attributes @a1@ and @a3@ of the object. Moreover, 
a message @m@, with
parameters @O@ and @w@, is read and consumed, and a new message @m'(O')@ is sent
\emph{with delay} @x@ (see~\cite{journ-rtm}).
``Irrelevant'' attributes, such as @a2@, need not be mentioned
in a rule.

 A \emph{flat} (or \emph{non-hierarchical}) object-oriented specification is one
 where all rewrites happen in the ``outermost'' configuration; that is, no attribute value
 $t$  rewrites to some $t'\not = t$.

The specification of time-dependent behavior of object-oriented real-time systems  follows the
 techniques given in~\cite{journ-rtm}.
 Time elapse is modeled by the tick rule

{\small
\begin{alltt}
var C : Configuration .   var T : Time .
crl [\(tick\)] : \texttt{\char123}C\texttt{\char125} => \texttt{\char123}delta(C, T)\texttt{\char125} in time T if T <= mte(C) [nonexec] .
\end{alltt}
}   

\noindent The function @delta@ defines the effect of time elapse 
on a configuration, and the function @mte@ defines the
maximum amount of time that can elapse before some action must take
place. These functions distribute over the objects and messages
 in a configuration
and must be defined for all single objects and messages to 
define the timed behavior of a system. The tick rule advances time
\emph{nondeterministically} by \emph{any} amount @T@ less than or equal to
@mte(C)@. To execute
such rules, Real-Time Maude offers a choice of
\emph{time sampling strategies}, so that only \emph{some} moments in time are
visited. The choice of such strategies includes:
\begin{itemize}
\item Advancing time by a fixed amount $\Delta$ in each application of
  a tick rule.
\item The \emph{maximal} strategy, that advances
 time to the next moment when some action must be taken, as defined by @mte@.
 This corresponds to \emph{event-driven
  simulation}. 
\end{itemize}

\paragraph{Formal Analysis.} 
A Real-Time Maude specification is \emph{executable}, under reasonable conditions, 
and the tool offers a variety of formal analysis 
methods. The \emph{rewrite} command simulates \emph{one} fair 
behavior of the system \emph{up to a certain duration}.
The \emph{search}
command uses a breadth-first strategy to  analyze all possible
  behaviors of the system, by checking whether a state matching a
  \emph{pattern} and satisfying
  a \emph{condition} can be reached from the
  initial state.  Such a pattern typically describes the \emph{negation} of an
invariant, so that the search succeeds iff the invariant is violated.
The  command which searches for $n$ states
satisfying the \emph{pattern} search criterion has syntax 

\small
\begin{alltt}
 (utsearch [\(n\)] \(t\) =>* \(pattern\) such that \(cond\) .)
\end{alltt}
\normalsize

Real-Time Maude also extends Maude's \emph{linear temporal logic model
  checker} 
  to check whether
each behavior, possibly  up to a certain time bound,
  satisfies a temporal logic 
  formula.
 \emph{State propositions} are terms of sort @Prop@, and their 
semantics should be 
given by (possibly conditional) equations of the form

\small
\begin{alltt}
  \texttt{\char123}\(statePattern\)\texttt{\char125} |= \(prop\) = \(b\) 
\end{alltt} 
\normalsize

\noindent for $b$ a term of sort @Bool@, which defines the state
proposition $prop$ to hold in all states $@{@t@}@$ where $@{@t@}@$
\verb+|=+ $prop$ evaluates to @true@.  We use the notation $\Pi$ for
the set of propositions and $L_\Pi$ for the (implicit) labeling
function assigning to each state the set of propositions that hold in the state.  A temporal
logic \emph{formula} is constructed by state propositions and the 
 Boolean and temporal logic operators discussed in Section~\ref{sec:mtl}.
  The time-bounded model checking command has syntax

\small
\begin{alltt}
  (mc \(t\) |=t \(\mathit{formula}\) in time <= \(\tau\) .)
\end{alltt}
\normalsize

\noindent for initial state $t$ and temporal logic formula $\mathit{formula}$ .

Since the model checking commands execute  
tick rules
according to the chosen time sampling strategy,  only a subset of all possible
behaviors is analyzed. Therefore, Real-Time Maude analyses are in general
\emph{incomplete} for a given property.
However, in~\cite{wrla06} we have given easily checkable conditions for 
ensuring that Real-Time Maude analyses are indeed sound and complete.

It is also worth remarking that in the rest of the paper, we implicitly consider
the different analyses w.r.t.\ Real-Time Maude executions. That is, for dense time, 
by ``a 
rewrite theory $\cal R$'' in the following sections we  typically mean 
the real-time rewrite theory $\mathcal{R}^{tss}$ that has been obtained
from an original time-nondeterministic real-time rewrite theory $\mathcal{R}$
by applying the theory transformation corresponding to using the time sampling strategy $tss$ 
when executing the tick rules~\cite{journ-rtm}.

%

\subsection{Metric Temporal Logic}\label{sec:mtl}

\emph{Linear temporal logic (LTL)}~\cite{pnueli:ltl}
allows us to describe properties of paths of a given system. 
The states are labeled with elements from a finite set $\Pi$ of atomic
propositions. Besides propositions and the usual Boolean
operators, LTL formulae can be built using the temporal \emph{until}
operator. Intuitively, the formula $p\ \until\ q$ (``$p$ until $q$'')
is satisfied by a path if the property $q$ becomes valid within an
arbitrary but finite number of steps and the property $p$ constantly
holds on the path before. As syntactic sugar we define
$\finally\ p$ (``eventually $p$'', defined as $\textit{true}\ \until\
p$) that is satisfied by a path if $p$ holds somewhere on the path, and
$\globally\ p$ (``globally $p$'', defined as $\neg (\textit{true}\
\until\ (\neg p))$) expressing that $p$ holds on the whole path. The
 \emph{weak until} operator $p\ \weakuntil\ q$ is 
defined as $(p\ \until\ q) \lor (\globally\ p)$.


For time-critical systems we need more expressive power to state that
some actions should happen \emph{within some time bounds}.  There are
different extensions of LTL to capture also timed properties
(see~\cite{real-time-logics} for an overview).  In this paper, we use
the extension \emph{metric temporal logic (MTL)}~\cite{koyMTL}, that
adds time interval bounds to the temporal operators.
For the until operator, the formula $p\ \until_{[t_1,t_2]}\
q$ states that $p\ \until\ q$ holds and, furthermore, $q$ occurs within
the time interval $[t_1,t_2]$. 
%

Formulae of MTL are built using the following abstract syntax:
\[
\varphi \quad ::= \quad \textit{true} \quad | \quad p \quad | 
\quad \neg \varphi \quad | \quad \varphi \land \varphi \quad | 
\quad \varphi\ \until_{[t_1,t_2]}\ \varphi
\]
with $p\in\Pi$ and either $t_1,t_2\in\reals$ with $t_1\leq t_2$ and
$t_2>0$, or $t_1\in\reals$ and $t_2=\infty$.  Note that
$\until_{[0,\infty]}$, for which we just write $\until$,
corresponds to the unbounded until of LTL.  Besides the
usual Boolean operators $\lor,\rightarrow,\ldots$ we define as
syntactic sugar $\finally_{[t_1,t_2]}\ \varphi$ as $\textit{true}\
\until_{[t_1,t_2]}\ \varphi$ and $\globally_{[t_1,t_2]}\ \varphi$ as
$\neg(\textit{true}\ \until_{[t_1,t_2]} (\neg \varphi))$.
If the
lower bound $t_1$ is $0$, we use the notation $\varphi_1\ \until_{\leq
  t_2}\ \varphi_2$, and analogously for the other operators.


Given a real-time rewrite theory $\cal{R}$, 
the set of \emph{states} is defined as $\mathbb{T}_{\Sigma/E,\texttt{GlobalSystem}}$.
A set $\Pi$ of (possibly
parametric) {\em atomic propositions} on those states can
be defined equationally in a protecting extension $(\Sigma \cup \Pi,E \cup D)
\supseteq (\Sigma,E)$, and give rise to a {\em labeling function} 
$L_{\Pi} : \mathbb{T}_{\Sigma/E,\texttt{GlobalSystem}} \rightarrow \mathcal{P}(\Pi)$ 
 in the obvious way~\cite{maude-book}.
%
%
Adapting the pointwise semantics for MTL given in~\cite{real-time-logics},
we can define satisfaction of MTL formulas for real-time 
rewrite theories over timed paths as follows: 
%



\begin{definition}
  Let $\cal{R}$ be  a real-time rewrite theory, $L_{\Pi}$ a labeling
  function on $\cal R$, and $\pi = t_0
  \stackrel{r_0}{\rightarrow} t_1 \stackrel{r_1}{\longrightarrow}
  \ldots $ a timed path  in  $\cal R$.  
  The satisfaction relation of an MTL formula $\phi$ for the
  path $\pi$ in $\cal R$
  is then defined recursively as follows:\\[1ex]
\begin{tabular}[h]{lp{11cm}}
  ${\cal R},L_{\Pi},\pi\models\textit{true}$ & always holds\\
  ${\cal R},L_{\Pi},\pi\models p$ & iff  $p\in L_{\Pi}(t_0)$\\  
  ${\cal R},L_{\Pi},\pi\models \neg \varphi$ & iff 
  ${\cal R},L_{\Pi},\pi\not\models\varphi$\\
  ${\cal R},L_{\Pi},\pi\models \varphi_1\land\varphi_2$ & iff 
  ${\cal R},L_{\Pi},\pi\models\varphi_1$ and ${\cal R},L_{\Pi},\pi\models\varphi_2$\\
  ${\cal R},L_{\Pi},\pi\models \varphi_1\ \until_{[r_a,r_b]}\ \varphi_2$ & 
  iff  there exists a $j\in\nats$ such that 
  ${\cal R},L_{\Pi},\pi^j\models\varphi_2$, \newline
  ${\cal R},L_{\Pi},\pi^i \models \varphi_1$ for all  
  $0\leq i < j$, and $r_a \leq \sum_{k=0}^{j-1}r_k \leq r_b$.
\end{tabular}\\[1ex]  
For a state $t_0$ of sort $\mathtt{GlobalSystem}$, the
satisfaction relation of an MTL formula $\phi$ for the state
$t_0$ in $\cal R$ is defined as:
\[
\mathcal{R},L_{\Pi},t_0 \models \phi \;\iff
\; \forall \pi\in Paths(\mathcal{R})_{t_0} \quad
\mathcal{R},L_{\Pi},\pi \models \phi
\]
\end{definition}

\section{Model Checking  MTL Properties of Object-Oriented Specifications}
\label{sec:implementation}

Real-Time Maude currently does not support MTL model checking. However, some
 MTL formulas can already be model checked in Real-Time Maude
using the \emph{time-bounded} search and LTL  model checking   commands. 
For example, we 
 can model check the
time-bounded 
until property $\mathcal{R},L_{\Pi},t_0
\models p \buntil q$, for $p$ and $q$ \emph{state} properties from $\Pi$,
 using the time-bounded model checking
command 

\small
\begin{alltt}
(mc t0 |=t \(p\) U \(q\) in time <= \(r\) .)
\end{alltt}
\normalsize
 
\noindent We can also analyze the properties $\mathcal{R},L_{\Pi},t_0 \models \balways p$
and $\mathcal{R},L_{\Pi},t_0 \models \beventually p$ in a similar way.

In this paper we present analysis algorithms for the following two
classes of MTL formulae:
\begin{enumerate}
\item \emph{Bounded response:}
  $\globally\ (p \rightarrow (\finally_{\leq r}\ q))$
\item \emph{Minimum separation:} $\globally\ (p \rightarrow
  (p\ \weakuntil\ (\globally_{\leq r}\ \neg p)))$
\end{enumerate}

We propose to transform an MTL model checking problem
$\mathcal{R},L_{\Pi}, t_0 \models \varphi$
 into an untimed LTL model
checking problem 
$\widetilde{\mathcal{R}},\widetilde{L}_{\Pi}, \widetilde{t_0} \models \widetilde{\varphi}$.
Both transformations add a \emph{clock} to the system: for model checking
bounded response properties, this clock measures the time since $p$
held without $q$ holding in the meantime; for minimum separation
properties, the clock measures the distance between two
non-consecutive $p$-states. We take care not to  increase the clocks 
``unnecessarily,'' so that if the state space reachable from 
$t_0$ in $\mathcal{R}$ is finite, then the state space reachable from 
$\widetilde{t_0}$ in $\widetilde{\mathcal{R}}$ remains finite, under reasonable 
time-divergence assumptions on the executions.

We assume that our specifications are \emph{tick-invariant}~\cite{wrla06}
with regard to the state propositions occurring in the formula, i.e.,
a tick step does not change the valuation of the atomic propositions
occurring in the formula.
Most systems, including the two case studies in the paper, 
satisfy tick-invariance, since the  state propositions usually do
not involve  the value of clock and timer attributes in the system.

%

\subsection{Bounded response: $\globally\ (p \rightarrow \finally_{\leq r}\ q)$}
\label{sec:bounded}

A bounded response property states that the system always reacts to a
request $p$ with an action $q$ within time $r$. For example, in our 
medical devices case study, 
the ventilation machine, helping a  sedated patient to breathe,
should not be stopped for more than two seconds at a time; that is,   each
state in which the machine is pausing must be followed by a state in 
which the machine is breathing in two seconds or less.

The MTL
model checking problem 
$$\mathcal{R},\LBR,t_0 \models \Box\: (p
\longrightarrow \Diamond_{\leq r}\: q)$$
\noindent for $p,q\in\Pi$ state
propositions, can be transformed into the untimed model checking problem
$$\widetilde{\mathcal{R}}_r,\tildeLBR,\widetilde{t_0} \models 
\Box\: (p\longrightarrow \Diamond\: q) \;\wedge\; 
\Box\:(\clockvaluesem(\cpq) \leq r)$$
\noindent where $\clockvaluesem(\cpq)$ is the value of a ``clock''
that measures the time since $p$ held without $q$ holding in the
meantime.  For real-time rewrite theories having only time-divergent
paths we could skip the first condition $\Box\: (p\longrightarrow
\Diamond\: q)$, that assures, that we also consider all
relevant time-convergent paths as possible counterexamples.

We add a ``clock'' $\cpq$ to
the system, and update it as follows:
\begin{itemize} 
\item[i)] If the clock $\cpq$ is turned off, and a state satisfying 
$p\wedge \neg q$ is reached, then the clock is set to 0 and is turned on.
\item[ii)] The clock is turned off when a state satisfying $q$ is reached.
\item[iii)] A clock that is on is increased according to the elapsed time
  in the system.
\end{itemize}

%

For the very useful class of ``flat'' object-oriented specifications
formalized according to the guidelines in~\cite{journ-rtm}---all advanced
Real-Time Maude applications have been so specified---we can automate
the transformation from $\mathcal{R},\LBR,t_0, p, q, r$ to 
$\widetilde{\mathcal{R}},\tildeLBR,\widetilde{t_0}$
as follows:
\begin{enumerate}
\item \label{lab:br_state} Add the following class for the clock:

\small
\begin{alltt}
 class Clock | \clockvalue : Time, status : OnOff .    
 sort OnOff .     ops on off : -> OnOff [ctor] .
\end{alltt}
\normalsize

\item \label{lab:br_init} Add a clock object to the initial state $\texttt{\char123}t_0\texttt{\char125}$, so that the initial state
 becomes 

\small
\begin{alltt}
\texttt{\char123}\(\tnull\)  <\! \(\cpq\)\! : Clock\! |\! \clockvalue\! :\! 0, status\! :\! \(x\)\! >\texttt{\char125}
\end{alltt}
\normalsize

\noindent where $\cpq$ is a constant of sort \texttt{Oid} and $x$ is @on@ if 
$p\in L(\texttt{\char123}t_0\texttt{\char125})$ and
$q\not\in L(\texttt{\char123}t_0\texttt{\char125})$, and is @off@ otherwise.
Note that $p\in L(\texttt{\char123}t_0\texttt{\char125})$ can be
checked in Maude by  checking whether
$\texttt{\char123}t_0\texttt{\char125}
@ |= @ p = @true@$. 
%
\item \label{lab:br_tick}  We keep Real-Time-Maude's object-oriented tick rule and extend
  the functions \texttt{delta} and \texttt{mte} to clocks as follows, ensuring that
  \texttt{mte} is not affected by the new clock object:

\small
\begin{alltt}
  eq delta(< \(\cpq\) : Clock | status : on, \clockvalue : T >, T') =
           < \(\cpq\) : Clock | \clockvalue : if T <= \(r\) then T + T' else T fi > .
  eq delta(< \(\cpq\) : Clock | status : off >, T') = < \(\cpq\) : Clock | > .
  eq mte(< \(\cpq\) : Clock | >) = INF .
\end{alltt}
\normalsize

\noindent Notice that the @delta@ function ensures that
the clock value never increases more than necessary, 
preserving \emph{finiteness} of the reachable state space from the initial state. 
\item \label{lab:br_inst} Each \emph{instantaneous} rule 
$t\; @=>@\; t' @ if @cond\;$ or $\;@{@ t@} => {@t'@} if @cond$
 in $\mathcal{R}$ is replaced
by the rules:

\small
\begin{alltt}
 \texttt{\char123}\(t\) REST < \(\cpq\) : Clock | status : on >\texttt{\char125}   
 =>  \texttt{\char123}\(t'\) REST < \(\cpq\) : Clock | >\texttt{\char125} if \texttt{\char123}\(t'\) REST\texttt{\char125} |= \(q\) =/= true and \(cond\)
\end{alltt}
\normalsize

\noindent (if the clock is on, then it continues to stay on if a state
satisfying $\neg q$ is reached);

\small
\begin{alltt}
 \texttt{\char123}\(t\) REST < \(\cpq\) : Clock | status : on >\texttt{\char125}  
 =>  \texttt{\char123}\(t'\) REST < \(\cpq\) : Clock | status : off >\texttt{\char125} if \texttt{\char123}\(t'\) REST\texttt{\char125} |= \(q\) and \(cond\) 
\end{alltt}
\normalsize

\noindent  (if the clock is on, then it is turned off when a state satisfying $q$ is reached);

\small
\begin{alltt}
 \texttt{\char123}\(t\) REST < \(\cpq\) : Clock | status : off >\texttt{\char125} 
 =>  \texttt{\char123}\(t'\) REST < \(\cpq\) : Clock | \clockvalue : 0, status : on >\texttt{\char125} 
     if  \texttt{\char123}\(t'\) REST \texttt{\char125} |= \(p\) and \texttt{\char123}\(t'\) REST\texttt{\char125} |= \(q\) =/= true and \(cond\) 
\end{alltt}
\normalsize

\noindent (if the clock is off, then it is set to 0 and turned on when a state 
	satisfying $p\wedge \neg q$ is reached);

\small
\begin{alltt}
 \texttt{\char123}\(t\) REST < \(\cpq\) : Clock | status : off >\texttt{\char125}   
 =>  \texttt{\char123}\(t'\) REST < \(\cpq\) : Clock | >\texttt{\char125}
     if  \texttt{\char123}\(t'\) REST\texttt{\char125} |= \(q\) or \texttt{\char123}\(t'\) REST\texttt{\char125} |= \(p\) =/= true and \(cond\)
\end{alltt}
\normalsize

\noindent (if the clock is off, then it continues to stay off if a state 
	satisfying $q\vee \neg p$ is reached).

In the above rules @REST@ is a variable of sort @Configuration@ that
 does not appear in the
original rule. @REST@ matches  the ``other'' objects
 and messages in the state.
\end{enumerate}

\noindent Summarizing, the \emph{\transBR-transformation}
transforms a real-time rewrite theory $\cal R$, a labeling function
$\LBR$ of $\cal R$ with $p,q\in\Pi$,  an initial state $t_0$ of $\cal R$, 
and a bounded response formula $ \Box\: (p
\longrightarrow \Diamond_{\leq r}\: q)$ into the
triplet $\widetilde{\cal R}$, $\tildeLBR$, and $\widetilde{t}_0$ by
\begin{itemize}
\item transforming $\cal R$ into $\widetilde{\cal R}$ according to the
  points \ref{lab:br_state}, \ref{lab:br_tick}, and \ref{lab:br_inst} above;
\item transforming $\LBR$ into $\tildeLBR$ by adapting its domain to
  the transformed state space, but letting the labeling otherwise
  unchanged, i.e., $\LBR(\{t\}) = \tildeLBR(\{t\ o\})$ for all states
  $t$ of $\cal R$ and all \texttt{Clock} instances $o$;
\item extending the initial state $t_0$ according to point
  \ref{lab:br_init} above, yielding $\widetilde{t}_0$.
\end{itemize}

\noindent The validity of the bounded response property $\globally\ (p \rightarrow
\finally_{\leq r}\ q)$ is equivalent to 
$\globally\ (p \rightarrow
\finally\ q)$ \emph{and}
the clock value being less than
or equal to $r$ in each reachable state of the transformed module. The latter property
can be defined as an atomic proposition

\small
\begin{alltt}
op clock`<=_ : Time -> Prop [ctor] .
eq \char123REST <\! \(\cpq\)\! :\! Clock\! |\! clock\! :\! T1\! >\char125 |= clock <= T2 = (T1 <= T2) .
\end{alltt}
\normalsize  

\noindent and hence bounded response can be analyzed using Real-Time Maude's untimed
LTL model checking features. We have implemented the above model transformation
in Real-Time Maude. We have also implemented a bounded response model checking
command in the tool based on this transformation. However, for pragmatic reasons,
we do \emph{not} model check the property $\widetilde{\mathcal{R}},\tildeLBR,\widetilde{t_0} \models 
\Box\: (p\longrightarrow \Diamond\: q) \;\wedge\; 
\Box\:(\clockvaluesem(\cpq) \leq r)$. Instead, we have observed   the unsurprising
fact that, with  time sampling strategy executions, all our large Real-Time Maude
applications are modeled as time-diverging theories. In these cases, bounded
response reduces to checking  $\widetilde{\mathcal{R}},\tildeLBR,\widetilde{t_0} \models 
\Box\:(\clockvaluesem(\cpq) \leq r)$, which  can be analyzed by the following search command
that searches for a state in which the clock value is greater than
$r$:

\small
\begin{alltt}
(utsearch [1] \texttt{\char123}\(\tnull\) <\! \(\cpq\)\! :\! Clock\! |\! \clockvalue\! :\! 0, status\! :\! \(x\)\! >\texttt{\char125} =>*
              \texttt{\char123}C:Configuration <\! \(\cpq\)\! :\! Clock | \clockvalue\! :\! T:Time\! >\texttt{\char125} such that T:Time > \(r\) .)
\end{alltt}
\normalsize

\noindent where $x$ is @on@ if $p\in L(\{t_0\})$ and $q\not\in
L(\{t_0\})$, and is @off@ otherwise.  The practical difference is
that, whereas the LTL model checking does not terminate when the state
space reachable from $t_0$ in $\mathcal{R}$ is infinite, the above
search command provides a \emph{semi-decision} procedure for the
invalidity of the bounded response property. For an example of the benefit of this
time-divergence-assuming implementation, consider the bounded response
analysis of the medical systems example in
Section~\ref{sec:case_studies}. The reachable state space is infinite
because of the clock used in the original model; hence any direct LTL
model checking would not terminate, but we see that our bounded
response command indeed returns a counterexample falsifying the
bounded response property.

In our tool, the bounded response model checking command (for the automatic 
\transBR-transformation and the execution the Real-Time Maude search) is written with
syntax

\small
\begin{alltt}
(br \(\tnull\) |= \(p\) => <>le( \(r\) ) \(q\) .)
\end{alltt}
\normalsize



\subsection{Minimum Separation: $\globally\ (p \rightarrow (p\ \weakuntil\
  \ \globally_{\leq r}\ \neg p))$}

Given a real-time rewrite theory $\cal R$ with a labeling function
$\LMS$,  $p\in\Pi$, all runs of $\cal R$
are made up of a sequence of blocks for which $p$ and $\neg p$ hold
alternatingly (see Figure~\ref{fig:chain}). The minimum separation
property requires that each $\neg p$-block occurring after a $p$-block
must have a minimum duration $r$.  I.e., if the run for which we check
the property starts with a $p$-block, then all $\neg p$-blocks of the
run must have a duration at least $r$.  Otherwise, if the run starts
with a $\neg p$-block, then the same holds for all $\neg p$-blocks
except the first one at the beginning of the run. 

  \begin{figure}[htb]
    \centering

\begin{tikzpicture}[]
    \node[] (p0) at (0,0) {};
    \node[] (p1) at (2,0) {};
    \node[] (p2) at (5,0) {};
    \node[] (p3) at (6,0) {};
    \node[] (p4) at (10,0) {};
    \node[] (p4b) at (9.5,0) {};
    \node[] (p5) at (13,0) {};
    \path[thick,|-|] (p0) edge node[below] {$\underbrace{\hspace*{1.7cm}}_{p}$} (p1)
                   edge node[below] {\hspace*{1.7cm}$\underbrace{\hspace*{3cm}}_{\neg p}$} (p2)
                   edge node[below] {\hspace*{4.7cm}$\underbrace{\hspace*{1cm}}_{p}$} (p3)
                   edge node[below] {\hspace*{5.7cm}$\underbrace{\hspace*{4cm}}_{\neg p}$} (p4);
    \path[thick,->] (p4b) edge node[below] {${\ldots}$} (p5);
    \node[] (pp1) at (1.75,0.3) {};
    \node[] (pp2) at (5,0.3) {};
    \node[] (pp3) at (5.75,0.3) {};
    \node[] (pp4) at (10,0.3) {};
    \path[<->] (pp1) edge node[above] {$\geq r$} (pp2);
    \path[<->] (pp3) edge node[above] {$\geq r$} (pp4);

\end{tikzpicture}

\begin{tikzpicture}[]
    \node[] (p0) at (0,0) {};
    \node[] (p1) at (1.5,0) {};
    \node[] (p2) at (4,0) {};
    \node[] (p3) at (6,0) {};
    \node[] (p4) at (9,0) {};
    \node[] (p5) at (11,0) {};
    \node[] (p5b) at (10.75,0) {};
    \node[] (p6) at (13,0) {};
    \path[thick,|-|] (p0) edge node[below] {$\underbrace{\hspace*{1.2cm}}_{\neg p}$} (p1)
                   edge node[below] {\hspace*{1.2cm}$\underbrace{\hspace*{2.5cm}}_{p}$} (p2)
                   edge node[below] {\hspace*{3.7cm}$\underbrace{\hspace*{2cm}}_{\neg p}$} (p3)
                   edge node[below] {\hspace*{5.7cm}$\underbrace{\hspace*{3cm}}_{p}$} (p4)
                   edge node[below] {\hspace*{8.7cm}$\underbrace{\hspace*{2cm}}_{\neg p}$} (p5);
    \path[thick,->] (p5b) edge node[below] {${\ldots}$} (p6);

    \node[] (pp2) at (3.75,0.3) {};
    \node[] (pp3) at (6,0.3) {};
    \node[] (pp4) at (8.75,0.3) {};
    \node[] (pp5) at (11,0.3) {};
    \path[thick,<->] (pp2) edge node[above] {$\geq r$} (pp3);
    \path[thick,<->] (pp4) edge node[above] {$\geq r$} (pp5);
\end{tikzpicture}

\caption{The form of runs satisfying the minimum separation property
  $\globally\ (p \rightarrow (p\ \weakuntil\ \ \globally_{\leq r}\ \neg
  p))$. The $p$- and $\neg p$-blocks may also be infinite.}
    \label{fig:chain}
  \end{figure}
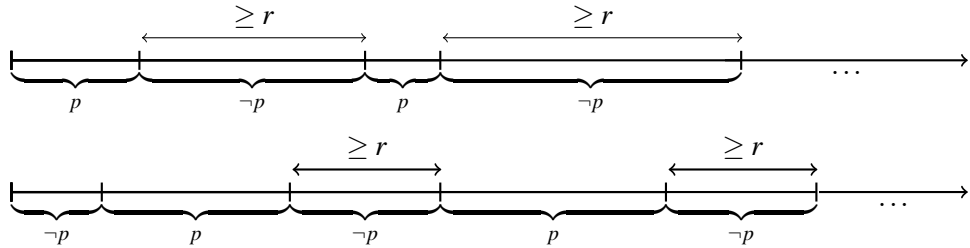

  We transform the MTL model checking problem
$$\mathcal{R},\LMS,t_0 \models \globally\ (p \rightarrow (p\ \weakuntil\ \  \globally_{\leq r}\ \neg p))$$
\noindent into the untimed model checking
problem
$$\widetilde{\mathcal{R}},\tildeLMS,\widetilde{t_0} \models 
\Box\:(status(\cp) = @on@ \;\vee\; \clockvaluesem(\cp) \geq r)$$
\noindent where $\clockvaluesem(\cp)$ is the value of a ``clock'' that
measures the time duration since we saw a $p$-state. That means, to model
check minimum separation properties, we add a ``clock'' $\cp$ to the
system, which is initially turned off and set to $r$: in this way we
ensure that an eventual initial $\neg p$-block does not cause a violation of the property.
We update the clock as follows:
\begin{enumerate} 
\item[i)] If we move from a $p$-state to a $\neg p$-state, then the clock is turned on and reset to 0.
\item[ii)] The clock is turned off when a state satisfying $p$ is reached.
\item[iii)] A clock that is on is increased according to the elapsed time
  in the system.
\end{enumerate}

We can automate the transformation to search for counterexamples of a minimum separation
property of the above form as follows:
\begin{enumerate}
\item \label{lab:ms_state} Add the same class for the clock  as in Section~\ref{sec:bounded}:   

\small
\begin{alltt}
 class Clock | \clockvalue : Time, status : OnOff .    
\end{alltt}
\normalsize

\item \label{lab:ms_init} Add a clock object to the initial state 
 $\texttt{\char123}t_0\texttt{\char125}$, yielding

\small
\begin{alltt}\small
\texttt{\char123}\(\tnull\) <\! \(\cp\)\! :\! Clock\! |\! \clockvalue\! :\! \(r\), status\! :\! off\! >\texttt{\char125}
\end{alltt}
\normalsize 

\noindent where $\cp$ is a constant of sort \texttt{Oid}. 

\item \label{lab:ms_tick} We keep Real-Time-Maude's object-oriented tick rule and extend
  the function \texttt{delta} and \texttt{mte} to clocks exactly 
  as in Section~\ref{sec:bounded}.   
\item  \label{lab:ms_inst} Each \emph{instantaneous} rule $t\; @=>@\;
 t' @ if @cond$ or $\;@{@ t@} => {@t'@} if @cond$
 in $\mathcal{R}$ is replaced
by the  rules:

\small
\begin{alltt}\small
 \texttt{\char123}\(t\) REST < \(\cp\) : Clock | status : on >\texttt{\char125}   
 =>  \texttt{\char123}\(t'\) REST < \(\cp\) : Clock | >\texttt{\char125} if \texttt{\char123}\(t'\) REST\texttt{\char125} |= \(p\) =/= true and \(cond\)
\end{alltt}
\normalsize

\noindent (if the clock is on, then it continues to stay on, if a state
satisfying $\neg p$ is reached);

\small
\begin{alltt}
 \texttt{\char123}\(t\) REST < \(\cp\) : Clock | status : on >\texttt{\char125}   
 =>  \texttt{\char123}\(t'\) REST < \(\cp\) : Clock | status : off >\texttt{\char125} if \texttt{\char123}\(t'\) REST\texttt{\char125} |= \(p\) and \(cond\)
\end{alltt}
\normalsize

\noindent (if the clock is on, then it is turned off when a state satisfying $p$ is reached);

\small
\begin{alltt}
 \texttt{\char123}\(t\) REST < \(\cp\) : Clock | status : off >\texttt{\char125}   
 =>  \texttt{\char123}\(t'\) REST < \(\cp\) : Clock | >\texttt{\char125} 
     if (\texttt{\char123}\(t\) REST\texttt{\char125} |= \(p\) =/= true or \texttt{\char123}\(t'\) REST\texttt{\char125} |= \(p\)) and \(\)\(cond\)
\end{alltt}
\normalsize

\noindent (the clock remains off, if either we are in a state satisfying $\neg p$ or we move
to a state satisfying $p$; the first condition is needed to avoid  switching the clock on in initial $\neg p$-blocks);

\small
\begin{alltt}
 \texttt{\char123}\(t\) REST < \(\cp\) : Clock | status : off >\texttt{\char125}   
 =>  \texttt{\char123}\(t'\) REST < \(\cp\) : Clock | status : on, \clockvalue : 0 >\texttt{\char125} 
     if \texttt{\char123}\(t\) REST\texttt{\char125} |= \(p\) and \texttt{\char123}\(t'\) REST\texttt{\char125} |= \(p\) =/= true and \(cond\)
\end{alltt}
\normalsize

\noindent (if the clock is off, and we move from a state satisfying $p$ to a 
state satisfying $\neg p$, then the clock is turned on and reset to 0).

  Again, \texttt{REST} is a variable of sort \texttt{Configuration} that
  does not appear in the original rule.  
\end{enumerate}

\noindent The \emph{\transMS-transformation} therefore 
transforms a real-time rewrite theory $\cal R$, a labeling function
$\LMS$  with $p\in\Pi$, an initial state $t_0$ of $\cal R$, a state proposition $p$, and a time value $r$
into the
triple $\widetilde{\cal R}$, $\tildeLMS$, and $\widetilde{t}_0$ by
\begin{itemize}
\item transforming $\cal R$ into $\widetilde{\cal R}$ according to the
  points \ref{lab:ms_state}, \ref{lab:ms_tick}, and \ref{lab:ms_inst} above;
\item transforming $\LMS$ into $\tildeLMS$ by adapting its domain to
  the transformed state space, but letting the labeling otherwise
  unchanged, i.e., $\tildeLMS(\{t\ o\}) = \LMS(\{t\})$ for all states
  $t$ of $\cal R$ and all \texttt{Clock}s  $o$;
\item extending the initial state $t_0$ according to point
  \ref{lab:ms_init} above, yielding $\widetilde{t}_0$.
\end{itemize}

\noindent Checking the minimum separation property $\globally\ (p
\rightarrow (p\ \weakuntil\ \ \globally_{\leq r}\ \neg p))$ is
equivalent to checking that the validity of $p$ implies that the clock
value is larger than or equal to $r$ in each state in the transformed
module. The violation of the latter can be checked by the following
search command that searches for a state in which the clock is off (which
implies that $p$ holds) and the clock value is smaller than $r$:

\small
\begin{alltt}
(utsearch [1] \texttt{\char123}\(\tnull\)  < \(\cp\) : Clock | \clockvalue : \(r\), status : off >\texttt{\char125} =>*
              \texttt{\char123}C:Configuration  < \(\cp\) : Clock | \clockvalue : T:Time, status : off >\texttt{\char125}
                such that T:Time < \(r\) .)
\end{alltt}
\normalsize

The above \transMS-transformation has been integrated in Real-Time
Maude,  and model checking  the
above minimum separation property can be done with the Real-Time Maude
command

\small
\begin{alltt}
(ms \(\tnull\) |= \(p\) separated by >= \(r\) .)
\end{alltt}
\normalsize

\iflong
\else
\section{Correctness of  Bounded Response Model Checking}
\label{sec:proof}

In this section we give the correctness proof for 
our bounded response model checking.  The correctness proof for minimum separation, 
which we omit due to lack of space, is quite similar, and can be found
in an extended version of this paper~\cite{MSproof}.

\fi

To increase readability, in the following we use the notation 
$\pi\models\phi$ instead of $\mathcal{R},L_{\Pi},\pi\models\phi$ if
$\cal R$ and $L_\Pi$ are clear from the context.

The following lemma states that the
\transBR-transformation only adds some observators to the original
systems, without modifying its behavior.

\begin{lem}
\label{lem:protectedBR}
Let $\cal R$ be a real-time rewrite theory,
$\LBR$ with $p,q\in\Pi$ a labeling function for $\cal R$, and
let $\{t_0\}$ be an initial state for $\cal R$. Let $\widetilde{\cal
  R}$, $\tildeLBR$, and $\{\widetilde{t}_0\}$ be the result of the
$\transBR$-transformation applied to $\cal R$, $\LBR$, and $t_0$.

Then
for each path
$
\{t_0\} \stackrel{r_0}{\rightarrow} \{t_1\}
\stackrel{r_1}{\rightarrow} \ldots
$
in $\cal R$ there is a path 
$
\{\widetilde{t}_0\} \stackrel{r_0}{\rightarrow} \{\widetilde{t}_1\}
\stackrel{r_1}{\rightarrow} \ldots 
$
in $\widetilde{\cal{R}}$ such that,
for all $i$, there exists $t_i'$ with $\widetilde{t}_i = 
t_i\ 
t_i'$ and vice versa,
for all paths
$
\{\widetilde{t}_0\} \stackrel{r_0}{\rightarrow} \{\widetilde{t}_1\}
\stackrel{r_1}{\rightarrow} \ldots 
$
in $\widetilde{\cal R}$ there is a path 
$
\{t_0\} \stackrel{r_0}{\rightarrow} \{t_1\}
\stackrel{r_1}{\rightarrow} \ldots 
$
in $\cal R$ such that, for all $i$, $\widetilde{t}_i = 
t_i\ 
t_i'$ for some $t_i'$.

\end{lem}
\begin{proof}
  Adding the clock class and a clock object to the initial state
  does not affect
  the original part of the state, and defining @mte@ of the additional clocks to be 
the infinity value @INF@ ensures that the new clocks don't modify the timed behavior of the 
(original) system.
  Furthermore, the transformation replaces each original rule by a
  number of new rules, such that (1) each new rule acts on the
  original state part as the original rule, and (2) for each original
  rule and each extended state to which the original rule is
  applicable there is exactly one new rule that is applicable. (1)
  assures that the new rewrites yield the same result for the original
  part of the state and (2) assures that no original paths are blocked by the new
  rules. Thus the transformation does not modify the original
  behavior.\\[2ex]
  ``$\rightarrow$'': Let $\{t_0\}
  \stackrel{r_0}{\rightarrow} \{t_1\}
  \stackrel{r_1}{\rightarrow} \ldots $ be a path of $\cal
  R$. We define 
  $$\widetilde{t}_i = 
  t_i 
\texttt{ < } \cpq \texttt{ : Clock | clock : }x_i\texttt{, status : }y_i \texttt{ >}$$

\noindent  for all $i$  with $\; x_i \in \mathbb{T}_{\cal{R},\texttt{Time}}\;$ and $\;y_i
  \in \mathbb{T}_{\cal{R},\texttt{OnOff}}\;$ given inductively as
  follows:
  \begin{itemize}
  \item $x_0 = 0$, and $y_0=\texttt{on}$ if $p \in \LBR(\{t_0\}) \land
    q \not\in \LBR(\{t_0\})$ and $y_0=\texttt{off}$ otherwise. 
  \item For all $i$, if there is a tick rule yielding the rewrite $\{t_i\}
    \stackrel{r_i}{\rightarrow} \{t_{i+1}\}$, then we distinguish
    between the following cases:
    \begin{itemize}
    \item If $y_i = \texttt{on}$ and $x_i\leq r$, then we define
      $y_{i+1} = \texttt{on}$ and $x_{i+1} = x_i + r_i$. \newline
      Note that with the
      definition of the \texttt{delta} equation we have
      $\{\widetilde{t}_i\}\stackrel{r_i}{\rightarrow}\{\widetilde{t}_{i+1}\}$.
    \item If $y_i = \texttt{on}$ and $x_i > r$, then we define
      $y_{i+1} = \texttt{on}$ and $x_{i+1} = x_i$. \newline
      Note that with the
      definition of  \texttt{delta}  we have
      $\{\widetilde{t}_i\}\stackrel{r_i}{\rightarrow}\{\widetilde{t}_{i+1}\}$.
    \item Else, if $y_i = \texttt{off}$, we define
      $y_{i+1} = \texttt{off}$ and $x_{i+1} = x_i$. \newline
      With the
      definition of  the \texttt{delta} equation we have
      $\{\widetilde{t}_i\}\stackrel{r_i}{\rightarrow}\{\widetilde{t}_{i+1}\}$.
    \end{itemize}
  \item For all $i$, otherwise there is an instantaneous rule
   $t\; @=>@\; t' @ if @cond\;$ or $\;@{@ t@} => {@t'@} if @cond$, yielding the
    rewrite $\{t_i\} \stackrel{r_i}{\rightarrow}\{t_{i+1}\}$ with $r_i
    = 0$.
    \begin{itemize}
    \item If $y_i = \texttt{on}$ and $\{t_{i+1}\} \texttt{ |= } q \texttt{ =/= true}$
      then we set $y_{i+1} = \texttt{on}$ and $x_{i+1}=x_i$. \newline
      Note that
      the first replacement of the original rule yields
      $\{\widetilde{t}_i\}\stackrel{r_i}{\rightarrow}\{\widetilde{t}_{i+1}\}$.
    \item If $y_i = \texttt{on}$ and $\{t_{i+1}\} \texttt{ |= } q$
      then we set $y_{i+1} = \texttt{off}$ and $x_{i+1}=x_i$. \newline
      Note that
      the second replacement of the original rule yields
      $\{\widetilde{t}_i\}\stackrel{r_i}{\rightarrow}\{\widetilde{t}_{i+1}\}$.
    \item If $y_i = \texttt{off}$, $\{t_{i+1}\} \texttt{ |= } p$, and
      $\{t_{i+1}\} \texttt{ |= } q \texttt{ =/= true}$ then we set
      $y_{i+1} = \texttt{on}$ and $x_{i+1}=0$. \newline
      Note that the third
      replacement of the original rule yields
      $\{\widetilde{t}_i\}\stackrel{r_i}{\rightarrow}\{\widetilde{t}_{i+1}\}$.
    \item Else, if $y_i = \texttt{off}$ and either $\{t_{i+1}\} \texttt{
        |= } q$ or $\{t_{i+1}\} \texttt{ |= } p \texttt{ =/= true}$ then
      we set $y_{i+1} = \texttt{off}$ and $x_{i+1}=x_i$. \newline
      Note that the fourth
      replacement of the original rule yields
      $\{\widetilde{t}_i\}\stackrel{r_i}{\rightarrow}\{\widetilde{t}_{i+1}\}$.
    \end{itemize}
  \end{itemize}
  Above we made use of the fact that by definition for each $i$, the
  corresponding labeling  $\tildeLBR(\{\widetilde{t}_i\})$
  in $\widetilde{\cal{R}}$ is equal to
  $\LBR(\{t_i\})$.  Clearly, all $\{\widetilde{t}_i\}$
  are states of $\widetilde{\cal R}$. Especially,
  $\{\widetilde{t}_0\}$ results from $\{t_0\}$ by the $\transBR$-transformation.  Thus
  $\{\widetilde{t}_0\} \stackrel{r_0}{\rightarrow}
  \{\widetilde{t}_1\} \stackrel{r_1}{\rightarrow} \ldots$ is a
  path of
  $\widetilde{\cal R}$.\\[2ex]
  ``$\leftarrow$'': Given a path $\{\widetilde{t}_0\}
  \stackrel{r_0}{\rightarrow} \{\widetilde{t}_1\}
  \stackrel{r_1}{\rightarrow} \ldots$ of $\widetilde{\cal R}$ such that
   $$\widetilde{t}_i = 
  t_i 
\texttt{ < } \cpq \texttt{ : Clock | clock : }x_i\texttt{, status : }y_i \texttt{ >}$$
\noindent for each $i$, we show that $\{t_0\}
  \stackrel{r_0}{\rightarrow} \{t_1\}
  \stackrel{r_1}{\rightarrow} \ldots $ is a path of $\cal
  R$.
  \begin{itemize}
  \item For all $i$, if
    $\{\widetilde{t}_i\}\stackrel{r_i}{\rightarrow}\{\widetilde{t}_{i+1}\}$
    can be gained by a tick rule in $\widetilde{\cal R}$, then clearly also 
    $\{t_i\}\stackrel{r_i}{\rightarrow}\{t_{i+1}\}$
    can be gained by a tick rule in $\cal R$.
  \item Otherwise if
    $\{\widetilde{t}_i\}\stackrel{r_i}{\rightarrow}\{\widetilde{t}_{i+1}\}$
    can be gained by an instantaneous rule in $\widetilde{\cal R}$, then
    the original rule which got replaced by the above one yields
    $\{t_i\}\stackrel{r_i}{\rightarrow}\{t_{i+1}\}$ 
    in $\cal R$.
  \end{itemize}
\end{proof}

The following lemma clarifies the semantics of the bounded response
property: On the one hand, if along a path after a $p$ event $r$ time
long no $q$ event occurs, then the path is a counterexample for the
property. On the other hand, if a path violates the bounded response
property, then either after a $p$ event $r$ time long no $q$ event
occurs, or the path is time-convergent and violates the unbounded
property $\globally\ (p\ \rightarrow (\finally\ q))$.
\begin{lem}
  \label{lem:br_eq}
  Let $\cal{R}$ be  a real-time rewrite theory, $\LBR$ with $p,q\in\Pi$ a labeling
  function for $\cal R$, and  $\pi = \{t_0\} \stackrel{r_0}{\rightarrow} \{t_1\}
  \stackrel{r_1}{\rightarrow} \ldots $ a path of $\cal R$.
Then 
\[
   \left[ \exists i,j.\ 0\leq i < j \ \land\ (\pi^i \models p) \
  \land\ \left(\forall i\leq k \leq j.\ \pi^k\not\models q\right) \ \land \
  \sum_{k=i}^{j-1}r_k > r\right] \quad \longrightarrow \quad 
\left[\pi \not\models 
  \globally\ (p\ \rightarrow (\finally_{\leq r}\ q)) \right] 
\]
and 
\[
\begin{array}{l}
  \left[ \pi \not\models 
  \globally\ (p\ \rightarrow (\finally_{\leq r}\ q))\right]
  \ \ \longrightarrow \ \
  \left[ \exists i,j.\ 0\leq i < j \ \land\ (\pi^i \models p) \
  \land\ \left(\forall i\leq k \leq j.\ \pi^k\not\models q\right) \ \land \
  \sum_{k=i}^{j-1}r_k > r\right]\\
  \phantom{\left[ \pi \not\models 
  \globally\ (p\ \rightarrow (\finally_{\leq r}\ q))\right]
  \ \ \longrightarrow \ \
  }
  \lor \ \left[\pi\not\models \globally\ (p\ \rightarrow (\finally\ q))\right]\ .
\end{array}
\]
\end{lem}

\begin{proof}
For the first implication, due to the semantics of MTL the following holds:
\[
\begin{array}{ll}
\exists i,j.\ 0\leq i < j \ \land\ (\pi^i \models p) \
  \land\ \left(\forall i\leq k \leq j.\ \pi^k\not\models q\right) \ \land \
  \sum_{k=i}^{j-1}r_k > r \ . & \rightarrow  \\ 
  \exists i.\ (\pi^i \models 
  p)\ \land\ 
  \forall j\geq i.\ \left(\sum_{k=i}^{j-1}r_k \leq r \rightarrow 
    \pi^j\not\models q\right) & \rightarrow \\
  \exists i.\ (\pi^i \models 
  p)\ \land\ 
  (\pi^i \not\models \finally_{\leq r}\ q) & \rightarrow \\
  \exists i.\ \pi^i \models 
  \neg (p\ \rightarrow 
  (\finally_{\leq r}\ q)) & \rightarrow \\
  \pi \not\models 
  \globally\ (p\ \rightarrow (\finally_{\leq r}\ q)) \ .
\end{array}
\]
For the other direction, 
\[
\begin{array}{ll}
  \pi \not\models 
  \globally\ (p\ \rightarrow (\finally_{\leq r}\ q)) 
  & \rightarrow \
  \exists i.\ \pi^i \models 
  \neg (p\ \rightarrow 
  (\finally_{\leq r}\ q)) \\
  & \rightarrow\
  \exists i.\ (\pi^i \models 
  p)\ \land\ 
  (\pi^i \not\models \finally_{\leq r}\ q) \\
  &  \rightarrow\
  \exists i.\ (\pi^i \models 
  p)\ \land\ 
  \forall j\geq i.\ \left(\sum_{k=i}^{j-1}r_k \leq r \rightarrow 
    \pi^j\not\models q\right) \ .
\end{array}
\]
Let $i$ be such an index with $\pi^i \models p$ and $\forall j\geq i.\
(\sum_{k=i}^{j-1}r_k \leq r \rightarrow \pi^j\not\models q)$. If
$\pi\not\models \globally\ (p\ \rightarrow (\finally\ q))$ then we are
ready. So assume $\pi\models \globally\ (p\ \rightarrow (\finally\
q))$, implying that there is a smallest index $l\geq i$ with
$\pi^l\models q$. From the above it follows that  $\sum_{k=i}^l r_k > r$.

Note that by definition $r>0$ and thus $l>i$. Let $j=l-1$.  From the
minimality of $l$ we first conclude that $\forall i\leq k \leq j.\
\pi^k\not\models q$. From the minimality of $l$ we furthermore
conclude that the rewrite
$\{\widetilde{t}_j\}\rightarrow\{\widetilde{t}_l\}$ is an
instantaneous step, and thus $\sum_{k=i}^j r_k = \sum_{k=i}^l r_k >
r$.  That means, 
\[
  \exists i,j.\ 0\leq i < j
\ \land\ (\pi^i \models p) \ \land\ \left(\forall i\leq k \leq j.\
  \pi^k\not\models q\right) \ \land \ \sum_{k=i}^{j-1}r_k > r \ .
\]

\end{proof}

The following main theorem formalizes the correctness of our transformation:
Firstly, if the bounded response property holds, then the model
checking algorithms will not provide any counterexample. Secondly, if
the bounded response model checking algorithm does not find any
counterexample, and if there are no time-convergent counterexamples,
then the property holds.
\begin{theorem}
\label{lem:extensionBR}
Let $\cal R$ be a real-time rewrite theory, $\LBR$ a labeling function
for $\cal R$ with $p,q\in\Pi$, and $\{t_0\}$ an initial state of $\cal
R$. Let $\widetilde{\cal{R}}$, $\tildeLBR$, and $\{\widetilde{t}_0\}$
be the result of the $\transBR$-transformation applied to $\cal R$,
$\LBR$, and $\{t_0\}$.  Then
  $$
  {\cal R},\LBR,\{t_0\} \models \globally\ (p\ \rightarrow
  (\finally_{\leq r}\ q)) \quad \longrightarrow \quad \widetilde{\cal
      R},\tildeLBR,\{\widetilde{t}_0\} \models \globally\
    (\clockvaluesem(\cpq) \leq r),
  $$
and 
  $$
  \widetilde{\cal
    R},\tildeLBR,\{\widetilde{t}_0\} \models (\globally\ (p \rightarrow (\finally\ q))) \land (\globally\
  (\clockvaluesem(\cpq) \leq r))
  \quad \longrightarrow \quad 
  {\cal R},\LBR,\{t_0\} \models \globally\ (p\ \rightarrow
  (\finally_{\leq r}\ q)) ,
  $$
  where $\clockvaluesem(\cpq)$ denotes the value of the @clock@
  attribute of the clock object $\cpq$.
\end{theorem}

\begin{proof}
  For the first statement we show that 
  \[
  \widetilde{\cal
    R},\tildeLBR,\{\widetilde{t}_0\} \not\models \globally\
  (\clockvaluesem(\cpq) \leq r)
  \]
  implies 
  \[
  {\cal R},\LBR,\{t_0\}
  \not\models \globally\ (p\ \rightarrow (\finally_{\leq r}\ q))\ .
  \]
  Thus assume $ \widetilde{\cal R},\tildeLBR,\{\widetilde{t}_0\}
  \not\models \globally\ (\clockvaluesem(\cpq) \leq r)$. That means,
  there exists a path $\widetilde{\pi} = \{\widetilde{t}_0\}
  \stackrel{r_0}{\rightarrow} \{\widetilde{t}_1\}
  \stackrel{r_1}{\rightarrow} \ldots $ of $\widetilde{\cal R}$ with
  $\widetilde{t}_i = 
  t_i 
\texttt{ < } \cpq \texttt{ : Clock | clock : }x_i\texttt{, status : }y_i \texttt{ >}$
 and a smallest index $j$ such that
  $x_j>r$. Since the clock value is initially $0$ and it increases
  only due to tick rules if the clock is on, the clock must have been
  switched on at some point before $j$. Furthermore, since $j$ is
  minimal, the clock is continuously on from the last point where it
  was switched on till $\widetilde{t}_j$.

  Assume $i<j$ to be the smallest index such that the clock is
  continuously on from $\widetilde{t}_i$ till $\widetilde{t}_j$.
  Either $i$ is $0$ and the initial state satisfies $p\land \neg q$
  and $x_i=0$, or $i>0$ and the rewrite from the $(i-1)$th state to
  the $i$th state switched the clock from off to on and reset it to
  $0$. In the latter case the corresponding rewrite has the condition
  that $p\land\neg q$ holds in the $i$th state. Thus $p\land\neg q
  \land x_i=0$ holds in state $\widetilde{t}_i$.  The clock was kept
  on from state $\widetilde{t}_i$ till state $\widetilde{t}_j$. The
  only rules yielding this behavior are the tick rules increasing the
  clock value with the duration of the rewrite, and instantaneous
  rules assuring the invariance of $\neg q$ and letting the clock
  value untouched.  Due to tick-invariance, tick rules cannot
  cause any change in the validity of the propositions, and $\neg q$ holds
  all the way from the $i$th till the $j$th state.  Furthermore, the
  clock value at state $j$ is the sum of the durations of the rewrites
  from the $i$th to the $j$th state.  Thus
  \[
  \exists i,j.\ 0\leq i < j \ \land\ (\widetilde{\pi}^i \models p) \
  \land\ \left(\forall i\leq k \leq j.\ \widetilde{\pi}^k\not\models
    q\right) \ \land \ \sum_{k=i}^{j-1}r_k > r
  \]  
  holds and with Lemma \ref{lem:br_eq} we get $ \widetilde{\pi}
  \not\models \globally\ (p\ \rightarrow (\finally_{\leq r}\ q))$.
  Using Lemma~\ref{lem:protectedBR} we conclude that there is also a
  path $\pi$ of $\cal R$ such that
  $
  \pi \not\models 
  \globally\ (p\ \rightarrow (\finally_{\leq r}\ q))
  $
  and thus 
  $
  {\cal R}, \LBR,\{t_0\} \not\models 
  \globally\ (p\ \rightarrow (\finally_{\leq r}\ q))$.\\[2ex]
  For the second statement assume that 
\[
{\cal R},\LBR,\{t_0\}
  \not\models \globally\ (p\ \rightarrow (\finally_{\leq r}\ q))
\]
  holds. We show that it implies 
\[
 \widetilde{\cal
    R},\tildeLBR,\{\widetilde{t}_0\} \not\models (\globally (p
  \rightarrow (\finally\ q))) \land (\globally\ (\clockvaluesem(\cpq)
  \leq r))\ .
\]
  Due to the assumption there exists a path $\pi = \{t_0\}
  \stackrel{r_0}{\rightarrow} \{t_1\} \stackrel{r_1}{\rightarrow}
  \ldots $ of $\cal R$ violating $\globally\ (p\ \rightarrow
  (\finally_{\leq r}\ q))$.  Now, either $ \widetilde{\cal
    R},\tildeLBR,\{\widetilde{t}_0\} \not\models \globally (p
  \rightarrow (\finally\ q))$ and we are ready, or due to
  Lemma~\ref{lem:protectedBR} there exists a path $\widetilde{\pi}
  = \{\widetilde{t}_0\} \stackrel{r_0}{\rightarrow}
  \{\widetilde{t}_1\} \stackrel{r_1}{\rightarrow} \ldots $ of
  $\widetilde{\cal R}$ also violating $\globally\ (p\ \rightarrow
  (\finally_{\leq r}\ q))$.  With Lemma~\ref{lem:br_eq} we get
  \[
  \exists i,j.\ 0\leq i < j \ \land\ (\widetilde{\pi}^i \models p) \
  \land\ \left(\forall i\leq k \leq j.\ \widetilde{\pi}^k\not\models
    q\right) \ \land \ \sum_{k=i}^{j-1}r_k > r.
  \]
  Let $i$ and $j$ be the smallest indices satisfying the above
  condition. 
  \begin{itemize}
  \item If $i=0$ then by the fact that $\widetilde{\pi}^i\models p
    \land \neg q$ we have by definition that the clock in $\widetilde{t}_0$ is on
    and has the value $0$.
  \item If $i>0$ and for all $n<i$, $\widetilde{t}_n$ does not satisfy
    $p\land\neg q$, then by definition of the initial state the clock is
    initially off and the clock does not get switched
    on until the $(i-1)$th state, thus the clock is off in the $(i-1)$th state.
  \item If $i>0$ and there is an $n<i$ with $\widetilde{t}_n$
    satisfying $p\land\neg q$, then from the minimality of $i$ we conclude that
    there is a minimal $n\leq m < i$ such that $\widetilde{t}_m$
    satisfies $q$. From the minimality of $m$ we conclude that
    $\{\widetilde{t}_{m-1}\}\stackrel{r_{m-1}}{\rightarrow}\{\widetilde{t}_m\}$
    is due to an instantaneous rule, which, by definition, switches
    the clock off.
  \end{itemize}
  Thus either $i=0$ and the clock is on in $\widetilde{t}_i$ with value $0$, or $i>0$
  and the clock is off in state $\widetilde{t}_{i-1}$.  Furthermore,
  in the latter case the $(i-1)$th state satisfies $\neg p \lor q$
  (otherwise $i$ would not be minimal), and the rewrite
  $\{\widetilde{t}_{i-1}\}\stackrel{r_i}{\rightarrow}\{\widetilde{t}_i\}$ is
  due to an instantaneous rule, which, again by definition, switches
  the clock on and resets its value to $0$.

  We get that the clock is on with value $0$ in $\widetilde{t}_i$.  As $\neg q$
  holds all the way from the $i$th till the $j$th state, the clock
  remains on from the $i$th till the $j$th state. The rewrites of
  $\widetilde{\cal R}$ assure that the clock value in state $\widetilde{t}_j$ is the
  duration $\sum_{k=i}^{j-1}r_k$ that is by assumption larger than
  $r$, what was to be shown.
\end{proof}

The following lemma states that finiteness of the state space is
preserved under the $\transBR$-transformation, implying that our bounded
response model checking algorithm terminates for finite-space systems.

\begin{lem}
\label{lem:finiteBR}
Given a real-time rewrite theory $\cal R$, a labeling function
$\LBR$ of $\cal R$ with $p,q\in\Pi$, an initial state $\{t_0\}$ of
$\cal R$, and a fixed time sampling strategy, and furthermore, assuming that
  \begin{itemize}
  \item there are only finitely many states reachable in $\cal R$ from
    initial state $\{t_0\}$ with the given time sampling, i.e., the
    set
\[
\{\{t_i\}\ |\ \pi = \{t_0\}\stackrel{r_0}{\rightarrow}\{t_1\}\stackrel{r_1}{\rightarrow}\ldots \ \in Paths(\mathcal{R})_{t_0},\ i\in\nats
    \}
\]
 is finite, and
\item the number of different rewrite durations in all possible paths
  in $\cal R$ from $\{t_0\}$ under the given time sampling is finite, i.e., the
  set
\[
\{r_i\ | \
    \pi = \{t_0\}\stackrel{r_0}{\rightarrow}\{t_1\}\stackrel{r_1}{\rightarrow}\ldots \ \in Paths(\mathcal{R})_{t_0},\ i\in\nats \}
\]
    is finite,
  \end{itemize}
  then the bounded response model checking algorithm for
 $\cal{R}$ using the same sampling strategy terminates.
\end{lem}
\begin{proof}
  Assume that the above conditions hold.  Notice that the bounded
  response model checking algorithm always terminates if the set of
  reachable states of the $\transBR$-transformation 
  (from its initial state and under the given time sampling) is finite.

  Since all instantaneous rules in the $\transBR$-transformation
  $\widetilde{\cal R}$ either leave the clock value untouched or reset
  the clock value to $0$, the finiteness of the state space is
  preserved under the instantaneous rules of $\widetilde{\cal R}$.
  For the tick rules, on the one hand, if the clock value gets larger
  than the bound $r$ in the bounded response formula, then the model
  checking algorithm finds a counterexample and thus terminates. On
  the other hand, since there are only finitely many possible rewrite
  durations, there are only finitely many possible clock values
  less than or equal to $r$. So if the clock value never exceeds $r$
  than the reachable state space of the $\transBR$-transformation remains
  finite and the algorithm terminates in this case, too.
\end{proof}

\section{Case Studies}
\label{sec:case_studies}

This section briefly presents two case studies where we 
use the new model checking commands. The analysis has been performed 
on a 2.4GHz Intel\textsuperscript{\textregistered} Core 2 Duo processor with
2 GB of RAM.

\subsection{A Network of Medical Devices}
We apply the new Real-Time Maude commands 
on a Real-Time Maude model of an interlock protocol for 
a small network or medical devices, integrating an X-ray machine, 
a ventilator machine,  and a controller. The example was proposed by Lui
Sha, and the Real-Time Maude model is explained in \cite{phuket08}.

The ventilator machine helps a sedated patient to breathe during a surgery. 
An X-ray can be taken during the surgery by pushing a button.
To allow an X-ray to be taken without blurring the picture,
the ventilator must be briefly turned off. Within a certain time bound, the X-ray
must be taken and then the ventilation machine must be restarted.
Furthermore, the ventilation machine should not be stopped too often. 
The model also addresses
nondeterministic message delays and 
clock \emph{drifts}.

 In this model, all events take place when some ``timer''
expires or when a message arrives. Therefore, as proved in~\cite{wrla06},   the
system can be analyzed  using the \emph{maximal}
 time sampling strategy which advances time until the  next timer expires, 
so that the analyses remain sound and complete. 
One time unit in the specification corresponds to one millisecond 
in the case study.

\paragraph{Bounded Response Analysis.}
One  requirement in this model is that 
``the ventilation machine should not pause for more
than two seconds at a time.'' This can be expressed by the 
 bounded response formula
\[ \Box \;(\mbox{``machine is pausing''} \;\longrightarrow\;
\Diamond_{\leq 2 sec} \:\mbox{``machine is breathing''}). \]

In order to analyze this property, we first define two 
state propositions, @isPausing@ and @isBreathing@, in the
expected way: @isPausing@ holds for states in which the ventilation machine
is not breathing, while @isBreathing@ holds when the ventilation machine is breathing.
%
%
%
The bounded response property is model checked using the
following Real-Time Maude command:

%




\small
\begin{alltt}
Maude> \emph{(br initState |= isPausing => <>le( 2000 ) isBreathing .)}
\end{alltt}
\normalsize

\noindent The result of this command is a path representing a counterexample
to the validity of the property:

\footnotesize
\begin{alltt}
Property not satisfied

Counterexample path:

\char123< ct : Controller | \clockvalue : 0, lastPauseTime : 0 >
< u : User | pushButtonTimer : 0, pushInterval : 60000 >
< vm : VentMachine | state : breathing >
< xr : X-ray | state : idle >\char125

=>[pushButton]

\char123< ct : Controller | \clockvalue : 0, lastPauseTime : 0 >
< u : User | pushButtonTimer : 60000, pushInterval : 60000 >
< vm : VentMachine | state : breathing >
< xr : X-ray | state : idle >
dly(pushButton,0,50,10)\char125 

=>[dlyMsgArrives]

...

=>[idle]

\char123< ct : Controller | \clockvalue : 44000/21, lastPauseTime : 3000 >
< u : User | pushButtonTimer : 1220000/21, pushInterval : 60000 >
< vm : VentMachine | state : stopBreathing(9000/7)>
< xr : X-ray | state : idle >\char125

=>[tick]

\char123< ct : Controller | \clockvalue : 11000/3, lastPauseTime : 3000 >
< u : User | pushButtonTimer : 170000/3, pushInterval : 60000 >
< vm : VentMachine | state : stopBreathing(0)>
< xr : X-ray | state : idle >\char125
\end{alltt}
\normalsize

\noindent The result shows that the bounded response requirement does not hold.
This is due to the fact that  the 
ventilation machine may pause for 2.22 seconds, since its
internal clock is a little slow (see \cite{phuket08}). 
A counterexample path is therefore produced, of which we display here only a part, showing the sequence of rules
that have been applied to reach a state where the clock
added internally to the system reaches a @clock@ value greater than 2000.
The analysis took less than a second to perform.

A similar analysis can be done to check whether the ventilation machine cannot pause 
for more than 2.5 seconds. Since this property holds, the execution of the 
bounded response command will simply not stop, since the 
state space reachable from the initial state is not finite (i.e. due to the controller @clock@ attribute, which
just increases as time advances).
 
\paragraph{Minimum Separation Analysis.}
Another requirement says that the ventilator cannot pause more than once in ten minutes. 
That is, the  \emph{minimum separation} between two pauses is ten minutes. This property can be model checked
in Real-Time Maude as follows:

%


\small
\begin{alltt}
Maude> \emph{(ms initState |= isPausing separated by >= 600000 .)}
\footnotesize
Property not satisfied

Counterexample path:

\char123< ct : Controller | \clockvalue : 0, lastPauseTime : 0 >
< u : User | pushButtonTimer : 0, pushInterval : 60000 >
< vm : VentMachine | state : breathing >
< xr : X-ray | state : idle >\char125

=>[pushButton]

...

=>[stopBreathing]

\char123< ct : Controller | \clockvalue : 5951000/9, lastPauseTime : 663000 >
< u : User | pushButtonTimer : 530000/9, pushInterval : 60000 >
< vm : VentMachine | state : stopBreathing(2000)>
< xr : X-ray | state : wait(2500/3)>\char125
\end{alltt}
\normalsize

\noindent The requirement does not hold and a counterexample path is produced in 
less than 10 secs, leading to a state where the internal
@Clock@ object reaches a @clock@  value smaller than $600000$, while its status is @off@.

\subsection{A Four-Way Traffic Intersection System}
In this section, we analyze a bounded response property of an object-oriented 
Real-Time Maude model of a distributed fault-tolerant 
four-way traffic light controller for cars and pedestrians
described in~\cite{traffic-light}.  
The traffic light system for the 4-way intersection is designed as a collection of autonomous concurrent
objects that interact with each other by asynchronous message passing. 
The system is highly parametric: ten different parameters can be specified for an initial state, such 
as the presence of failures or emergency vehicles in the environment. 
Each 4-way intersection has two roads crossing in two directions: 
east-west (@EW@ in the specification) and north-south (@NS@ in the specification).
Each road has its own traffic lights. Each pedestrian light has a button that 
can be pushed by a pedestrian in order to get the green light and cross the street.
The behavior of the four-way intersection is as expected.

We  focus on the requirement
 that ``no pedestrian should wait for more than five minutes'' to cross a road. This corresponds to 
the bounded response formula
\[ \Box \;(\mbox{``pedestrian pushes the button''} \;\longrightarrow\;
\Diamond_{\leq 5 min} \:\mbox{``pedestrian light is green''}). \]
In order to analyze this property, we use the state propositions @buttonPushed@
and @pedLightGreen@ that take as parameter the direction of the crosswalk.
In  less than 3 minutes, we successfully verified that the pedestrian 
does not have to wait for more than 15 time units by executing the following Real-Time Maude command
(a time unit corresponds
to 15 seconds):

\small
\begin{alltt}
Maude > \emph{(br init("Imoan", minGreenTime + 2, minRedTime, 0, 0, 0, 1, 1, false, 0) 
            |= buttonPushed(NS) => <>le( 15 ) pedLightGreen(NS) .)}
          
Property satisfied
\end{alltt}
\normalsize

\noindent Furthermore, executing the same command, but for 14 time units, returned a counterexample.

\section{Related Work}
\label{sec:related}

There are several works determining decidable fragments of timed
temporal logics (e.g., \cite{bouyer:phd,quaknine:mtl}) in order to support
model checking algorithms for real-time systems.  The tools
\textsc{Kronos}~\cite{kronos} and \textsc{REDLIB}~\cite{wang:redlib}
are two TCTL (timed CTL) model checkers for timed automata. The
popular timed-automaton-based tool
\textsc{Uppaal}~\cite{uppaalTutorial} provides model checking only for
a ``reachability subset'' of TCTL that does not include bounded
response or minimum separation.

The contrast to our work is already explained in the introduction. 
Whereas the timed automaton formalism is quite restrictive
for the exact purpose of achieving  decidability of analyses,
Real-Time Maude, and even  its flat object-oriented subset considered in this paper,
 is a much more expressive model. The cost of this expressiveness is of course
that most properties are in general undecidable for Real-Time Maude. So also for 
the model checking commands in this paper, which are not guaranteed to terminate for
many Real-Time Maude models. Furthermore, since for dense time, Real-Time Maude
executes the tick rules according to a  time sampling strategy, we must also
prove that, even when terminating, our model checking analyses are both sound and complete, 
using, e.g.,  the techniques in~\cite{wrla06}.
Another obvious difference is that
we are covering only a fairly small, but important,  subset of a MTL.  


\newcommand{\ignore}[1]{}

\ignore{
In this paper we deal with model checking MTL properties for real-time
systems. Decidability of MTL over finite timed words was shown
in~\cite{quaknine:mtl}. First steps towards robust model-checking of
nested MTL properties for timed automata is described
in~\cite{Bouyer:robustmodel-checking}.  An overview on linear-time
temporal logics for real-time systems, covering subsets and
extensions of MTL and including some decidability and undecidability
results (for the pointwise and for the continuous semantics), can be
found in \cite{bouyer:phd}.  The complexity of model-checking formulas
of MTL, MITL, and TCTL over restricted sets of timed paths is the
topic of~\cite{raskin:restricted}.

Related work on model checking timed properties of real-time systems
include \emph{timed-automata based} tools for TCTL, a timed extension
of CTL.  

\textsc{Kronos}~\cite{kronos} is a verification tool for real-time
systems developed at \textsc{Verimag}. In \textsc{Kronos}, components
of real-time systems are modeled by timed automata.  Correctness
requirements are expressed in TCTL. The model checking algorithm
implemented in \textsc{Kronos} for checking TCTL for timed automata is
based upon a \emph{symbolic} representation of the infinite state
space by sets of linear constraints.

\textsc{Uppaal}~\cite{uppaalTutorial} is an integrated tool
environment for modeling, validation, and verification of real-time
systems. Systems are specified as networks of timed automata, extended
with data types (bounded integers, arrays, etc.). The tool offers
different alternatives for state space representation, e.g., using
difference bound matrices, under-, and over-approximations. It
provides an efficient \emph{symbolic} model checking procedure for a
subset of non-nested TCTL properties. Nested TCTL properties are not
supported by \textsc{Uppaal}.

\textsc{REDLIB}~\cite{wang:redlib} is a library constructed out of the
TCTL model checker \textsc{RED}. It uses BDD-like diagrams for the
efficient representation and manipulation of dense-time state-spaces.
Full TCTL model checking is provided for timed automata.
}

\section{Concluding Remarks}
\label{sec:conclusion}

This paper has explained how we have enriched 
the important class of flat object-oriented Real-Time Maude models 
with model checking features  for bounded response and
minimum separation properties. 

Object-oriented Real-Time Maude 
specifications capture many systems
that cannot be specified as timed automata; indeed, all advanced
Real-Time Maude applications have been so specified. It is therefore
not  surprising that the model checking problems we address are undecidable
in general. Therefore, our model checking analyses may fail to terminate, although they will
terminate if the properties do \emph{not} hold.  Furthermore, 
our model checking commands are executed with a selected time sampling strategy,
so  that only a subset of all possible behaviors are analyzed. Hence, 
our analyses may be incomplete or unsound. Nevertheless, for object-oriented
specifications we have identified easily checkable conditions that ensure
soundness and completeness of (untimed) model checking. Further on the positive
side, we have shown that (with reasonable assumptions on the treatment of dense time),
our model checking analyses terminate when the reachable state space is finite. 

The implementation of our  model checking procedures 
follows a transformational approach that takes advantage of Maude's high
performance search command by transforming an MTL model
checking problem into checking the validity of an invariant property. 
We proved the correctness of these transformations under mild conditions,
such as tick-invariance and time divergence. 

The  model checking commands have been integrated into Real-Time Maude and 
have been successfully used to model check 
a small network of medical devices~\cite{phuket08}, as well as on a larger model of a
traffic intersection system~\cite{traffic-light}.

The present work is just our first foray into 
model checking metric temporal logic properties for Real-Time Maude specifications. 
Much work remains ahead. First of all, we should extend the class
of MTL formulas we can model check, and extend the classes of
Real-Time Maude models for which such model checking can be performed. 
For example, if the present techniques could be extended to 
\emph{non-flat} (or \emph{hierarchical} ``Russian dolls'') object-oriented
Real-Time Maude specifications, then we would get for free model checkers
for these properties for both behavioral AADL models and hierarchical Ptolemy II DE models.
We should also extend the commands to analyze only paths up to a certain duration, 
so that the reachable state space becomes finite. 
The correctness proofs in this paper all deal with correctness w.r.t.\ the executed paths.
We must of course further investigate the soundness and completeness of
such analyses w.r.t.\ all possible behaviors of a system.  

\paragraph{Acknowledgments.} We thank the anonymous reviewers for very 
helpful comments on a previous version of this paper, and 
gratefully acknowledge financial support by the Research Council of Norway
through the Rhytm project, and by the Research Council of Norway and the 
German Academic Exchange Service (DAAD) through the DAADppp project 
 "Hybrid Systems Modeling and Analysis
with Rewriting Techniques (HySmart)."

\ignore{
The general transformational approach for model checking MTL properties, 
presented in this paper, can be applied to other classes of 
MTL properties. We plan for the future to extend this approach to 
other MTL properties, such as classes of nested bounded until 
(i.e., $(a\ \until_{\leq r_b}\ b)\ \until_{\leq r_c}$ ), or properties
that are required by other case studies. It would be also
interesting to extend the bounded response and
minimum separation commands to general interval 
bounds on the temporal operator appearing in the formula.
Furthermore, since the model checking procedures are available
in Real-Time Maude as new commands, it will be easy and 
interesting to test them on more case studies. 
}

\ifworking
{\begin{center} \Large \bf Still \thetodo\ things to fix! \end{center}}
\fi

\bibliographystyle{eptcs} 
\bibliography{bib/string,bib/papers,bibl,bib/crossref}

\iflong
\else
%


\fi

\end{document}